\documentclass[conference]{IEEEtran}
\makeatletter
\def\ps@headings{%
\def\@oddhead{\mbox{}\scriptsize\rightmark \hfil \thepage}%
\def\@evenhead{\scriptsize\thepage \hfil \leftmark\mbox{}}%
\def\@oddfoot{}%
\def\@evenfoot{}}
\makeatother 

\usepackage{amsmath,dsfont,yfonts,amssymb,epsfig,amsthm,bm,multirow, graphicx,color}
\usepackage{stfloats, url}

\newtheorem{theorem}{Theorem}
\newtheorem{proposition}[theorem]{Proposition}
\newtheorem{corollary}[theorem]{Corollary}

\newtheorem{prop}{Proposition}[section]
\newtheorem{cor}{Corollary}
\newtheorem{lm}{Lemma}
\newtheorem{thm}{\bf Theorem}
\newcommand{\bthm}{\begin{thm}}
\newcommand{\ethm}{\end{thm}}

\newcommand{\bcor}{\begin{cor}}
\newcommand{\ecor}{\end{cor}}
\newcommand{\bprop}{\begin{prop}}
\newcommand{\eprop}{\end{prop}}
\newcommand{\blm}{\begin{lm}}
\newcommand{\elm}{\end{lm}}
\newcommand{\beq}{\begin{equation}}
\newcommand{\eeq}{\end{equation}}
\newcommand{\ber}{\begin{eqnarray}}
\newcommand{\eer}{\end{eqnarray}}

\newenvironment{proof1}{\begin{trivlist}\item[]{\bf Proof:\hspace{2mm}}}{\hfill$\blackbox$\end{trivlist}}

\newcommand{\blackbox}{\vrule height7pt width5pt depth1pt}

\newcommand{\bit}{\begin{itemize}}
\newcommand{\eit}{\end{itemize}}
\newcommand{\ben}{\begin{enumerate}}
\newcommand{\een}{\end{enumerate}}
\newcommand{\bdesc}{\begin{description}}
\newcommand{\edesc}{\end{description}}
\newcommand{\beqarrn}{\begin{eqnarray*}}
\newcommand{\eeqarrn}{\end{eqnarray*}}
\newenvironment{proofof}[1]{\begin{trivlist}\item[]{\bf Proof of #1:\hspace{2mm}
}}{\hfill\blackbox\end{trivlist}}
\newcommand{\bproofof}{\begin{proofof}}
\newcommand{\eproofof}{\end{proofof}}
\newenvironment{rem}{\begin{trivlist}\item[]{\bf
Remark:}\hspace{4mm}}{\end{trivlist}}
\newcommand{\brem}{\begin{rem}}
\newcommand{\erem}{\end{rem}}
\newenvironment{rems}{\begin{trivlist}\item[]{\bf
Remarks}\begin{itemize}}{\end{itemize}\end{trivlist}}
\newcommand{\brems}{\begin{rems}}
\newcommand{\erems}{\end{rems}}
\newtheorem{fact}{Fact}
\newcommand{\bfact}{\begin{fact}}
\newcommand{\efact}{\end{fact}}
\newtheorem{examp}{Example}
\newcommand{\bexamp}{\begin{examp}\rm}
\newcommand{\eexamp}{\end{examp}}
\newtheorem{defn}{Definition}
\newcommand{\bdefn}{\begin{defn}\rm}
\newcommand{\edefn}{\end{defn}}

\newtheorem{prob}{Problem}
\newcommand{\bprob}{\begin{prob}}
\newcommand{\eprob}{\end{prob}}

\newcommand{\bvtm}{\begin{verbatim}}
\newcommand{\bfig}{\begin{figure}}
\newcommand{\efig}{\end{figure}}
\newcommand{\bcen}{\begin{center}}
\newcommand{\ecen}{\end{center}}

\long\def\comment#1{}

\def \n2{{N_0 \over 2}}

\def \h5{\hspace{0.5in}}

\begin{document}

\title{Multiple Timescale Dispatch and Scheduling for Stochastic Reliability in
Smart Grids with Wind Generation Integration}

\author{\IEEEauthorblockN{Miao He and Sugumar Murugesan and Junshan Zhang}
\IEEEauthorblockA{School of Electrical, Computer and Energy
Engineering, Arizona State
University, Tempe, AZ 85287\\
Email:\{Miao.He, Sugumar.Murugesan, Junshan.Zhang\}@asu.edu }}
\date{}
\maketitle

\begin{abstract}
Integrating volatile renewable energy resources into the bulk power
grid is challenging, due to the reliability requirement that at each
instant the load and generation in the system remain balanced. In
this study, we tackle this challenge for smart grid with integrated
wind generation, by leveraging multi-timescale dispatch and
scheduling. Specifically, we consider smart grids with two classes
of energy users - traditional energy users and opportunistic energy
users (e.g., smart meters or smart appliances), and investigate
pricing and dispatch at two timescales, via day-ahead scheduling and
real-time scheduling. In day-ahead scheduling, with the statistical
information on wind generation and energy demands, we characterize
the optimal procurement of the energy supply and the day-ahead
retail price for the traditional energy users; in real-time
scheduling, with the realization of wind generation and the load of
traditional energy users, we optimize real-time prices to manage the
opportunistic energy users so as to achieve system-wide reliability.
More specifically, when the opportunistic users are non-persistent,
i.e., a subset of them leave the power market when the real-time
price is not acceptable, we obtain closed-form solutions to the
two-level scheduling problem. For the persistent case, we treat the
scheduling problem as a multi-timescale Markov decision process. We
show that it can be recast, explicitly, as a classic Markov decision
process with continuous state and action spaces, the solution to
which can be found via standard techniques.

We conclude that the proposed multi-scale dispatch and scheduling
with real-time pricing can effectively address the volatility and
uncertainty of wind generation and energy demand, and has the
potential to improve the penetration of renewable energy into smart
grids.


\end{abstract}


\section{Introduction}

\subsection{Motivation}
To address the grand challenge of a sustainable energy industry,
there has recently been a surge of interest in alternative energy
resources, including wind, solar, bio-fuel, and geothermal energy.
Ultimately, all these energy solutions hinge heavily on smart grid
technologies that are capable of coordinating and managing
dynamically interacting power grid participants. There is therefore
an urgent need to develop a new generation of cyber-enabled energy
management system (EMS) and supervisory control and data acquisition
(SCADA), which can carry out reliable and possibly distributed
management of these energy sources.

For normal operations of power systems, the precise balance between
energy supply and demand is of the most significance to system
reliability. Integrating a large amount of intermittent renewable
energy resources (e.g., wind generation) into the bulk power grid
has put forth great challenges for generation planning and system
reliability. In particular, a complication that arises is that some
primary elements of renewable energy sources, such as wind and solar
power, are highly variable (stochastic) and often uncontrollable,
making it difficult to guarantee that the load and generation in the
system remain balanced at each instant. A mismatch between supply
and demand could cause a deviation of zonal frequency from nominal
value \cite{Chow05}, and when it gets severe power outages and
blackouts may occur. Further, wind generation is non-dispatchable,
in the sense that the output of wind turbines must be taken by all
rather than by request. Moreover, the volatility and intermittence
makes it difficult for system operators to obtain accurate knowledge
of future wind generations. Traditionally, system operators maintain
additional generation capacity and reserves (on-line or fast-start),
at additional costs, to address the supply uncertainty.

While the volatility of wind generation induces uncertainties in the
supply side of the power grid, an emerging class of energy users,
namely \textit{opportunistic energy users}, induce uncertainties on
the \textit{demand} side as well.
It is noted in \cite{Newman08} that over $10\%$ daily energy
consumption in United States is from the usage of appliances such as
water heater, cloth dryers, and dish washers, which are envisaged to
become \textit{smart} and be branded as opportunistic energy users,
with the following behaviors distinct from \textit{traditional
energy users}: 1) they access the energy system in an opportunistic
manner, according to the availability of system resources; 2)
different from the `always-on' demand of traditional energy users,
the load profiles of opportunistic energy users can be bursty and
can be either inelastic or elastic; 3) opportunistic energy users
respond to the power market on a much finer timescale.


The prevalence of the new class of opportunistic energy users, if
utilized intelligently, makes demand side management (DSM) a
promising solution to reduce the costs incurred by the deep
penetration of wind generation \cite{Sioshansi10a}. Traditionally,
some energy users (e.g., residential and small commercial users) pay
a fixed price per unit of electricity that is established to
represent an average cost of power generation over a given
time-frame (e.g., a month), independent of the generation cost. In
contrast, under price-based DSM programs \cite{Albadi07}, the retail
prices are tied with the generation cost and may vary according to
the availability of energy supplies. Often times, the energy market
consists of a day-ahead market and a real-time market for
electricity. Simply put, the day-ahead  market produces financially
binding schedules for the energy generation and consumption  one day
before the operating day. Further, the real-time market is used to
tune the balance between the energy amount scheduled day-ahead and
the real-time load. However, it is known that existing dynamic
pricing mechanisms (term of use, critical peak pricing, etc) do not
work well for handling generation uncertainty and managing the
demand uncertainty in a real-time manner (i.e., within minutes).

In this study, we will explore multi-timescale dispatch and
scheduling to address the following challenges: 1) the supply
uncertainty as a result of the volatility and non-stationarity of
wind generation; 2) the demand uncertainty due to a large number of
opportunistic energy users and their stochastic behaviors; 3) the
coupling between sequential decisions across multiple timescales.


\subsection{Summary of Main Contributions}

Aiming to tackle the challenge of integrating volatile wind
generation into the bulk power grid, we study dispatch and
scheduling, for a smart grid model with two classes of energy users,
namely traditional energy users and opportunistic energy users
(e.g., smart meters or smart appliances). We consider a power grid
with both conventional energy sources (e.g., thermal) and wind
generation. Notably, wind generation is among the renewable
resources that has most variability and uncertainty, and exhibits
multi-level dynamics across time. To enhance the penetration of wind
energy, we study multi-timescale dispatch and scheduling based on a
marriage of real-time pricing and  multi-settlement power market
economics.

Specifically, the system controller performs scheduling at two
timescales. In the day-ahead schedule, with the statistical
information on wind generation and energy demands, the operator
optimally procures conventional energy supply and decides the
optimal retail price for the traditional energy users, for the next
day. In the real-time schedule, upon the realization of the wind
energy generation and the demand from traditional energy users
(which is stochastically dependent on the day-ahead retail price),
the controller decides the real-time retail price for the
opportunistic energy users.

In particular, we explore multi-scale scheduling for two types of
opportunistic energy users: the non-persistent and the persistent
users. The non-persistent users leave the power market when they
find that the current real-time price is unacceptable, whereas the
persistent opportunistic users wait for the next acceptable
real-time price. We obtain closed-form solutions for the scheduling
problem when the users are non-persistent. For the persistent case,
the scheduling problem is a multi-timescale Markov decision process
(MMDP) that we recast, explicitly, as a standard Markov decision
process that can be solved via standard solution techniques. We
demonstrate, via numerical experiments, that the proposed
two-timescale dispatch and scheduling enables the penetration of
wind generation and hence improves the overall efficiency, by
enabling two-way energy exchange between providers and customers,
and by facilitating both information interaction as well as energy
interaction.


\subsection{Related Work}
Roughly speaking, related work falls into two major categories: the
scheduling of power systems with wind generation integration; and
the pricing and management of opportunistic users.

Very recent work \cite{VaraiyaTalk} proposed \textit{risk-limiting
dispatch} in contrast to {\em worst-case dispatch} for traditional
generation planning, by treating scheduling with known demand and
uncertain supply as a multi-stage decision problem with recourse.
With the same spirit, \cite{Bouffard08} proposed \textit{stochastic
security} criterion in scheduling with wind generation.
Specifically, a scenario-based approach (a scenario is defined as
the trajectory of potential realization of future supply) is used
to minimize the expected aggregated social cost incurred by all the
scenarios in the scheduling horizon.
\textit{Model predictive dispatch} of wind generation was studied in
\cite{Le08}, which treats future wind generation in the
 scheduling horizon as a dynamic process specified by an ARMA model, aiming to minimize the
 total generation cost subject to the physical constraints of conventional generators.
In a nutshell, all the works noted above focus on managing the
energy supply and treat the energy demand as known (or statistically
known) but uncontrollable.

Real-time pricing in related works \cite{Albadi07,Sioshansi10a}
takes place on the timescale of hours in conventional power systems,
and the price response of energy users are understood with some
high-level models. For a cyber-enabled energy system with a large
number of opportunistic energy users, clearly it is more desirable
to accomplish the price response and manage the energy demand in a
much finer timescale (e.g. minutes). Indeed, this is one salient
feature envisioned for  future smart grids, although this poses
significant challenges for planning, modeling, and controlling
energy generation, transmission and distribution.

The remainder of the paper is organized as follows. In Section~II,
we describe the energy system model and give a brief overview of the
two-timescale settlement power market. We study two-timescale
dispatch and scheduling when the opportunistic users are
non-persistent in Section~III. In Section~IV, we consider persistent
opportunistic users and formulate the scheduling problem as a
multi-timescale Markov decision process. We provide concluding
remarks and identify directions for future research in Section~V.

\section{System Model and Problem Formulation}
\begin{figure}[htb]
\begin{center}
\includegraphics
[width=0.35\textwidth]{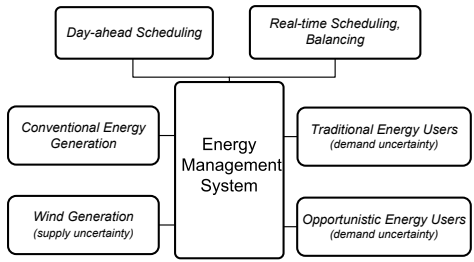}
\end{center}
\caption{An energy management system for smart grid with wind
generation and real-time pricing}\label{fig:Structure}
\end{figure}



We divide a 24-hour period into $M$ slots of length $T_1$ each,
where $M>1$; and each $T_1$-slot, in turn, consists of $K$ slots,
each of length $T_2$, with $K> 1$. For example, the $T_1$-slots can
correspond to hours in a day and the $T_2$-slots can correspond to
minutes.
We begin with an introduction to the supply side of the power
system.
\subsection{Energy Sources and Generation Costs}


We consider a power grid with two kinds of energy sources: the
conventional energy (e.g., thermal) and the wind energy. The
conventional energy is, in turn, drawn from two sources: base-load
generation and peaking generation, with generation cost $c_1$ and
$c_2$ per unit, respectively\footnote{The wind energy is assumed to
be cost-free, based on \cite{wood96}.}. Peaking generation is
typically from fast-start generators (e.g., gas turbines), with a
higher generation cost ($c_2 > c_1$). Due to the start-up time and
ramp rate of generators, the base-load generators are scheduled
day-ahead (24-hours ahead of the corresponding time) for each $T_1$
slot of the next day, and the generation cost $c_1$ contains
start-up cost $c_p$ and other operating costs. In real-time
scheduling of each $T_2$ slot, peaking generation and wind
generation are used, as needed, to clear the balance between demand
and the base-load generation.

Let $D$ denote the aggregate energy demand (from both traditional
and opportunistic users) and $W$ be the wind generation amount, in a
$T_2$-slot. Let $L=D-W$ denote the net demand for conventional
energy in a $T_2$ slot. Suppose that the system operator schedules a
base-load generation amount of $S$ for a $T_1$-slot and thus
schedules $s = {S \over K}$ for each $T_2$-slot within the
$T_1$-slot, for the next day. In real-time scheduling, the system
operator need to balance demand and supply: 1) If $L\ge s$, then
system operator dispatches a base-load generation of $s$ and a
peaking generation of $L-s$; 2) Otherwise, the over-scheduled
generation incurs a start-up cost of $c_p$ per unit.

\subsection{Day-Ahead Pricing and Real-Time Pricing}
As is standard, dynamic pricing contracts between the system
operator and end-users are used to manage the demand of both
traditional energy users and opportunistic energy users, on both
$T_1$ and $T_2$ timescales. We consider the following two-timescale
pricing model:
\begin{itemize}
\item Traditional energy users and opportunistic energy users
have separate pricing contracts: day-ahead pricing for traditional
users and real-time pricing for opportunistic users, respectively;
\item \textbf{Day-ahead retail price:} Traditional energy users are informed, one day ahead, of the day-ahead prices $u$ corresponding to each $T_1$-slot;
\item \textbf{Real-time retail price:} opportunistic energy users receive the real-time retail prices $v$ at the
beginning of each $T_2$-slot;
\item \textit{Both} retail prices have price cap $u_{cap}$ and $v_{cap}$, respectively.
\end{itemize}
Note that besides the day-ahead pricing model for the traditional
users, which was proposed in \cite{Albadi07}, we take a
forward-looking perspective to identify a new class of users, namely
the opportunistic energy users, and devise a real-time pricing model
for these users that is cognizant of the uncertainties in the demand
and supply. We discuss these uncertainties next.

\subsection{Supply Uncertainty of Wind Generation}


Wind generation is determined by the geographical and meteorological
conditions, and may assume high fluctuations or relative steady
patterns at different time-scales. Therefore, wind generation are
generally considered to be non-stationary and volatile.
%
%
%
Based on recent works \cite{Hetzer08}, \cite{xu04}, the wind
generation can be modeled as a non-stationary Gaussian random
process across $T_1$-slots, i.e., the wind generation amount in
$k$th $T_2$-slot of $m$th $T_1$-slot is given by
\begin{equation}
W_{k,m} \sim \mathcal{N}(\theta_m, \sigma^2),
\end{equation}
where $\theta_m$ is the mean and $\sigma^2$ is the variance. The
statistical information is available, one day ahead, to the
day-ahead scheduler. Indeed, this information is commonly provided
by the forecasting functions of EMS or commercial entities.

\subsection{Demand Uncertainty under Real-time Pricing}
In a $T_2$-slot, the aggregated demand from energy users, i.e., $D$,
consists of two components: $D_t$
from the traditional energy users, and $D_o$ from the opportunistic energy users. \vspace{10pt}\\
{\bf Demand model for traditional energy users:} The demands of
traditional energy users could be known with reasonable accuracy,
and the short-term price response of demand is well understood
\cite{Allcott09,Boisvert04}. Based on \cite{Fleten05}, we model the
energy demand of traditional energy users $D_t$ in a $T_2$-slot as a
random variable with mean depending on the day-ahead price $u$:
\begin{equation}
D_t = \mathbf{E}\left[D_t\right]+ \varepsilon_t,\label{eq:Dp}
\end{equation}
\noindent with
\begin{equation}
\mathbf{E}\left[D_t\right] =  \alpha_t {u }^{\gamma_t},
\label{eq:PriceElast}
\end{equation}
where $\varepsilon_t$ accounts for the uncertainty of $D_t$ and
$\gamma_t$ is the price elasticity of the traditional energy users
at the corresponding $T_1$-slot. Price elasticity is traditionally
used to characterize the price response of energy users and is
formally defined as the ratio of marginal percentage change in
demand to that of price, i.e., {\small
\begin{equation}
\gamma_t = {u \over D_t} {dD_t \over du}.
\end{equation}}\noindent
Note that, for power systems, the price elasticity is negative, and
$\alpha_t$ is the normalizing constant. Worth noting is that
$\alpha_t$ and $\gamma_t$ could be different for each of the
$T_1$-slots, since traditional energy users' demand and persistence
for consumption may depend on the term of the day (noon or evening,
peaking or off-peaking hours).\vspace{10pt}\\
{\bf Demand model for opportunistic energy users:} Under real-time
pricing, we assume that opportunistic energy users have the
following behaviors:
\begin{itemize}
\item Opportunistic energy users arrive according to a Poisson process
with rate $\lambda_0$, which is constant within a $T_1$-slot but can
vary across the $T_1$-slots;
\item Opportunistic energy users choose the access time randomly and independently;
\item In each $T_2$-slot, an opportunistic user $i$ arriving in the system decides to accept or reject the announced real-time price $v$ by comparing with a price acceptance level $V_i$. This price acceptance level
    is randomly chosen and is \textit{i.i.d} across the opportunistic users. Thus with $N$ denoting the total number of opportunistic energy users in a $T_2$-slot, the number of \textit{active} opportunistic users, $N_{a}$, is given by
    \begin{equation}
N_a  = \sum\limits_{i = 1}^N {\mathbf{1}_{\{V_i \ge v
\}}};\label{eq:ActUserNo}
\end{equation}
\item Each active opportunistic energy user has a per-unit energy consumption of $E_o$. Thus the total energy demand from opportunistic energy users is given by $D_o=N_a
E_o$;
\item Under real-time pricing, the response of opportunistic energy users may vary
according to the applications. Recent study \cite{Allcott09}
suggests that households respond to high energy prices through
energy conservation with no load shifting, while \cite{Boisvert04}
finds that most of the commercial customers respond by load shifting
to a later time. With these insights, we consider two kinds of
opportunistic energy users: non-persistent and persistent, of which
the former leaves the system if the real-time price is unacceptable,
while the latter waits in the system for a new real-time price in
the next $T_2$-slot.
\end{itemize}

\subsection{General Problem Formulation}
\begin{figure}[htb]
\begin{center}
\includegraphics
[width=0.45\textwidth]{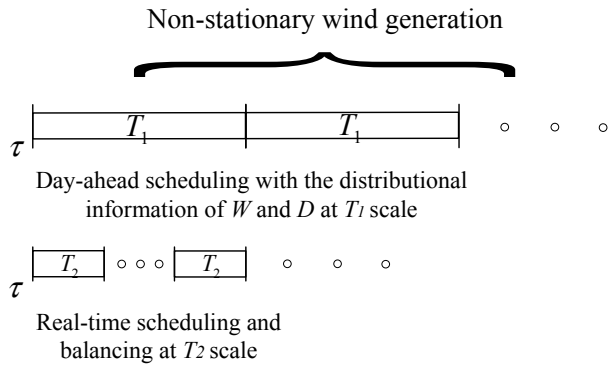}
\end{center}
\caption{Multiple timescale scheduling with non-stationary wind
generation and real-time pricing}\label{fig:Two-scale}
\end{figure}

As illustrated in Fig.~\ref{fig:Two-scale}, we consider day-ahead
scheduling and real-time scheduling with non-stationary wind
generation over a horizon of $M$ $T_1$-slots, with the main
objective being to maximize the overall expected profit. We
elaborate further on multi-timescale dispatch and scheduling below.

%

In day-ahead scheduling, with the distributional information of
next-day wind generation, the system operator aims to find a policy
$\bm{\pi}$ that dictates the two-timescale decisions $S$, $u$, and
$v$. A general formulation of the two-timescale scheduling problem
is provided below:
\begin{equation}
\mathbf{\mathcal{P}:\;\;} \mathop {\max }\limits_{\bm{\pi}}
\sum_{m=1}^{M} {R_m^u ( \bm{\pi})}, \label{eq:formDA}
\end{equation}
where, $R_m^u( \bm{\pi})$ is the total profit in a $T_1$-slot under
policy $\bm{\pi}$, given by
\begin{equation}
R_m^u(\bm{\pi}) = \sum\limits_{k = 1}^K E^{\bm{\pi}}_{\psi^l_{k,m}}
R_{k,m}^l (\psi^l_{k,m}, \bm{\pi}),
\end{equation}
where $R_{k,m}^l$ is the net profit in the $k$th $T_2$-slot of the
$m$th $T_1$-slot (henceforth called the $(k,m)$th slot) and
$\psi^l_{k,m}$ is the system state in the $(k,m)$th slot, which is
\textit{observable} in real time. When the opportunistic users are
non-persistent, $\psi^l_{k,m}$ consists of the wind energy and the
energy demand from traditional users. When the opportunistic users
are persistent, $\psi^l_{k,m}$ consists of wind generation and
traditional users' energy demand, as well as the energy requests of
opportunistic users carried over from the previous $T_2$-slot.

\section{Dispatch and Scheduling with Non-persistent Opportunistic Energy Users}

Note that the energy procurement and retail price in day-ahead
schedule have significant impact on the real-time retail price. The
real-time pricing policy, in turn, affects the optimization of the
day-ahead schedule. This tight coupling, underscores the need for
joint optimization of the day-ahead and real-time schedules. To this
end, we take a ``bottom-up'' approach in formulating the two-level
scheduling problem. Specifically, we first formulate the real-time
scheduling problem, conditioned on the day-ahead scheduling
decisions $S$ and $u$; and then introduce the day-ahead scheduling
problem while taking into account the real-time scheduling policy.
%

\subsection{Real-time Scheduling on Timescale $T_2$}
Given the amount of conventional energy procurement $S$ and the
day-ahead price $u$ settled in day-ahead scheduling, together with
the realizations of wind generation $W$ and traditional energy
demand $D_t$, the real-time scheduling problem in a $T_2$-slot is
formulated as:
\begin{equation}
\mathbf{\mathcal{P}^{RT}_{non-pst}:~~} \mathop {\max }\limits_{v}
R^l (\psi^l, s,u,v),
\end{equation}
\noindent where $\psi^l = \{W, D_t\}$ is the system state and
\begin{eqnarray}\label{eq:9}
R^l(\psi^l,s,u,v)&=&u
D_{t}+E_{D_o}^v[vD_o+\textbf{1}_A(-c_p s)\nonumber\\
&&+~\textbf{1}_B(-c_p\epsilon - c_1(s-\epsilon))\nonumber\\
&&+\left.\textbf{1}_C(-c_1 s+c_2\epsilon)\right],
\end{eqnarray}
where, recall that $D_o$ denotes the energy demand of opportunistic
energy users. The quantity $\epsilon$ denotes the surplus energy,
given by
\begin{eqnarray}
\epsilon&=&W+s-(D_{t}+ D_{o}).
\end{eqnarray}
The indicator function $\mathbf{1}_A$ corresponds to the event when
the wind energy is sufficient to meet the demands of both types of
users. Indicator $\mathbf{1}_B$ refers to the event when the wind
energy is not sufficient but the total of the scheduled traditional
energy and the wind energy is higher than the aggregate demand
$D_{t}+D_{o}$, which necessitates the cancelation of part of the
scheduled generation from the base-line generators incurring a
penalty $c_p$ per unit for the canceled amount of generation.
Indicator $\mathbf{1}_C$ corresponds to the event that the total of
scheduled baseline generation and the wind generation is
insufficient to meet the aggregate energy demand and the controller
must purchase the deficit energy from fast start-up generators at a
cost $c_2$. Formally, the events are described below:
\begin{equation}
\mathbf{1}_{A}= \left\{
\begin{array}{ll}
1  & \mbox{if}~ W\ge D_{t}+D_{o}\\
0  & \mbox{otherwise}
\end{array}\right.
\end{equation}
\begin{equation}
\mathbf{1}_{B}= \left\{
\begin{array}{ll}
1  & \mbox{if}~ \textbf{1}_{A}=0~ \mbox{and}~ \epsilon\ge 0\\
0  & \mbox{otherwise}
\end{array}\right.
\end{equation}
\begin{equation}
\mathbf{1}_{C}= \left\{
\begin{array}{ll}
1  & \mbox{if}~ \epsilon<0\\
0  & \mbox{otherwise}
\end{array}\right.
\end{equation}

\subsection{Day-ahead Scheduling on Timescale $T_1$}
The day-ahead scheduling problem in the non-persistence case is
formally given by
\begin{eqnarray}
\mathbf{\mathcal{P}^{DA}_{non-pst}}: \mathop {\max }\limits_{S,u}
\sum_{k=1}^{K}{\mathbf{E}_{W_k} \mathbf{E}_{D_{t_k}}^{u} \max_{v_k}
\left[ R^l_k (\psi^l_k, s, u, v_k) \right]}. \label{eq:formDA1}
\end{eqnarray}
where the inner maximization corresponds to the real-time scheduling
studied above. Since $W_k$ and $D_{t_k}$ are independent and
identically distributed across the $T_2$-slots, the day-ahead
scheduling problem in the non-persistent case can be optimized by
simply considering the snapshot problem in a specific $T_2$-slot,
given by\footnote{We drop the suffix $k$ for notational
simplicity.}:
\begin{eqnarray}
\mathbf{\mathcal{P}^{DA}_{non-pst}}: \mathop {\max }\limits_{S,u}
{\mathbf{E}_{W} \mathbf{E}_{D_{t}}^{u} \max_{v} \left[ R^l(\psi^l,
s, u, v) \right]}. \label{eq:formDA2}
\end{eqnarray}

\subsection{Approximate Solutions}
The tight coupling between the day-ahead and real-time scheduling
problems, along with the convolved nature of the uncertainties
involved, makes a direct joint optimization challenging. We
therefore take an alternate approach and obtain approximate
solutions to the real-time schedule. Based on the approximate
real-time schedule, we propose solutions to the day-ahead schedule.
In light of the characteristics of practical systems, we impose the
following general conditions first.
\begin{itemize}
\item \textbf{Condition $\mathbf{A}$:} Wind generation is not sufficient to
meet the total energy demand in the system.
\end{itemize}

Needless to say, under condition $\mathbf{A}$, $1_A=0$ in
(\ref{eq:9}). Define $Y = s + W - D_{t}$ as the energy procurement
for the demand of opportunistic energy users, then (\ref{eq:9})
reduces to
\begin{equation*}
R^l (\psi^l, s, u, v) = \mathbf{E}^v_{D_o} \left[ (v -c_2) D_o  - c
\mathbf{1}_B (Y-D_o) \right] ~~
\end{equation*}
\begin{equation}
~+~ u D_t - c_1 s + c_2 Y, \label{eq:formRT0}
\end{equation}
\noindent where $c \buildrel \Delta \over = c_p-c_1+c_2$.

We begin with an analysis of the real-time scheduling problem.
Recall that the number of non-persistent opportunistic energy users
arriving in the $T_2$-slot has a Poisson distribution, with mean
$\kappa_1 = \lambda_o T_2$. It follows that the number of users that
become active, denoted as $N_a$, is also a Poisson random variable
with mean $\kappa_1 \mathbf{P}(V \ge v)$.

Note that the price elasticity of the opportunistic energy users,
$\gamma_o$ is given by
\begin{equation}
\gamma_o = {v \over N_a} {dN_a \over dv}. ~~~~\label{eq:ElasDef}
\end{equation}

The opportunistic energy users are said to be \textit{relatively
inelastic} if $-1 \le \gamma_o < 0 $, i.e., the percentage change in
demand is greater than that of price; otherwise, they are
\textit{relatively elastic}. Further, let $v_{min}$ be the highest
price that is acceptable to all opportunistic energy users.

Then, the solution to (\ref{eq:ElasDef}) is given by:{
\begin{equation}
N_a = N \left(v / v_{min} \right)^{\gamma_o}. \label{eq:N_a}
\end{equation}}\noindent

Based on (\ref{eq:ActUserNo}), it is easy to see
\begin{equation}
\mathbf{P}(V \ge v) \approx \alpha_o
v^{\gamma_o},\label{eq:PersUserElas}
\end{equation}
\noindent where $\alpha_o  {\buildrel \Delta \over
=}v_{min}^{-\gamma_o}$ is the normalizing constant.

It is known \cite{Newman08} that, $\kappa_1$ is typically large in
practical power systems, and hence, $N_a$ could be approximated by a
Gaussian random variable. Accordingly, the demand of opportunistic
energy users $D_o = N_a E_s$ follows a Gaussian distribution
$\mathcal{N} (q_o(v), \sigma_o^2(v))$, with
\begin{equation}
\nonumber q _o(v) \buildrel \Delta \over = \kappa_1 \alpha_o
v^{\gamma_o} E_o,
\end{equation}
\begin{equation}
\sigma_o^2(v) \buildrel \Delta \over = \kappa_1 \alpha_o
v^{\gamma_o} E_o^2. \label{eq:DsElas}
\end{equation}

Plugging (\ref{eq:DsElas}) in (\ref{eq:formRT0}), we obtain
\begin{eqnarray}
~~\nonumber \tilde{R}^l (\psi^l, s, u, v) =  u { D_t} - c_1 s + c_2
Y + \left(v -c_2\right) q_o(v)
\\
\nonumber  -~ {\left(2\pi\right) }^{-1/2} c \sigma_o(v) \exp\left(-{
y^2/2 }\right)
~~\\
-~ c \left( Y - q_o(v) \right) \left(1 - Q\left( y \right)\right),
~~~~~\label{eq:formRT1}
\end{eqnarray}

\noindent where $ y \buildrel \Delta \over = {{Y-q_o(v)} \over
\sigma_o(v)}$. Denote the real-time pricing policy that solves the
preceding equation by the mapping $\mathbf{\tilde{\vartheta}}_{s,u}:
\psi^l \rightarrow v$, where, recall that $\psi^l$ is the system
state $\{W, D_t\}$.

Using the statistical properties of the wind generation and the
opportunistic user demand, reported in literature, we now further
simplify $\tilde{R}^l (\psi^l, s, u, v)$. According to
\cite{Agabus09}, the standard deviation of wind generation, is of
the same order as the mean of wind generation. Typically, wind
generation and the demand of opportunistic energy users are
comparable. Observe from (\ref{eq:DsElas}) that the variance of the
demand of opportunistic energy users is of the same order as its
mean. We conclude that $\frac{\sigma}{\sigma_o}$ has the same order
as the square root of the average demand of opportunistic users.
Thus, with $\sigma_Y^2 = \sigma_t^2 + \sigma^2$, it follows that
$\sigma_Y^2 \gg \sigma_o^2$.

Recall that the opportunistic energy users are said to be
\textit{relatively inelastic} if $-1 \le \gamma_o < 0 $, i.e., the
percentage change in demand is greater than that of price;
otherwise, they are \textit{relatively elastic}. Since the price
elasticity can have significant impact on $\sigma_o^2$ and $
\sigma_Y^2$, we proceed to study real-time schedule for different
cases of elasticity, with $\sigma_Y^2 \gg \sigma_o^2$.

\begin{proposition}
Suppose \textbf{condition} $\mathbf{A}$ holds. When the
non-persistent opportunistic energy users are relatively inelastic,
i.e., $-1\le \gamma_o<0$, the real-time pricing policy is given by
$\mathbf{\tilde{\vartheta}}_{s,u}(\psi^l)=v_{cap}$. \label{prop:1}
\end{proposition}

\begin{proof}
In practical power systems, according to \cite{Fleten05,Agabus09},
$D_t$ and $W$ usually have continuous, symmetrical and unimodal
probability distributions. Since $\sigma_Y^2 \gg \sigma_o^2$,
intuitively, there exists a finite constant $c_0 \ge 0$, such that:
\begin{equation*}
\mathbf{P}\left( \{Y - q_o(v) \le - c_0 \sigma_o(v)\} \right.
~~~~~~~~~~~~~~~~~~~~~~
\end{equation*}
\begin{equation}
~~~~\left. \cup \{Y - q_o(v) \ge c_0 \sigma_o(v)\} \right) \approx
1,
\end{equation}
\noindent and\footnote{It is well-known that (\ref{eq:ApproxC}) is
valid for $c_0 \ge 3$.}
\begin{equation}
Q(-c_0) \approx 1,~~ Q(c_0) \approx 0,~~\exp(-c_0^2/2) \approx
0.\label{eq:ApproxC}
\end{equation}
If $Y - q_o(v) \ge c_0 \sigma_o(v)$, (\ref{eq:formRT1}) boils down
to:
\begin{equation*}
\tilde{R}^l(\psi^l, s, u, v) =  \left(v - (c_1 - c_p)\right) q_o(v)
+ (c_1 - c_p) Y ~~~~
\end{equation*}
\begin{equation}
+~ u D_t - c_1 s; \label{eq:Approx1} ~~~~~~~~~
\end{equation}
\noindent When $Y - q_o(v) \le - c_0 \sigma_o(v)$,
(\ref{eq:formRT1})
simplifies to{
\begin{equation}
\tilde{R}^l(\psi^l, s, u, v)  =   \left(v - c_2\right) q_o(v) + u
D_t - c_1 s + c_2 Y.  \label{eq:Approx2}
\end{equation}}\noindent
It is clear that (\ref{eq:Approx1}) and (\ref{eq:Approx2}) are
unimodal for $v \in [v_{min}, v_{cap}]$, both with peaks at
$v=v_{cap}$. This yields the real-time pricing policy:
$\mathbf{\tilde{\vartheta}}_{s,u}(\psi^l)=v_{cap}$.
\end{proof}

\textbf{Remarks}: Note that the result in the preceding proposition
is intuitive, since, with the opportunistic users' energy demand
being relatively insensitive (inelastic) to the announced real-time
price, the scheduler can maximize profit by simply announcing the
highest possible price, $v_{cap}$.

\begin{proposition} \label{prop:2}
Suppose \textbf{condition} $\mathbf{A}$ holds. When the
non-persistent opportunistic energy users are relatively elastic,
i.e., $\gamma_o< -1$, the real-time pricing policy
$\mathbf{\tilde{\vartheta}}_{s,u}$ is given by
\begin{equation}
~~~~\mathbf{\tilde{\vartheta}}_{s,u}(\psi^l) = \left\{
\begin{array}{ll}
{\gamma_o(c_1 - c_p) \over {1+\gamma_o}} & {~~~~ \text{if}~~ Y  \ge
q_o({{\gamma_o}(c_1 - c_p) \over {1+\gamma_o}})} \cr { {\gamma_o}
c_2 \over {1+\gamma_o}} & {~~~~ \text{if}~~ Y < q_o({{\gamma_o} c_2
\over {1+\gamma_o}})} \cr {q_o^{-1}(Y)} & {~~~~ \text{o.w.}
}\nonumber
\end{array}\right.
\end{equation}
\end{proposition} \begin{proof}
When $\gamma_o< -1$, $\sigma_o(v)$ can be expected to be much
smaller than that in the inelastic case under the same real-time
prices. With this observation, we resort to the \textit{certainty
equivalence} techniques \cite{Water81}. By approximating $D_o$ with
its mean $q_o(v)$, the profit in real-time scheduling is given by:
\begin{equation}
\nonumber \tilde {R}^l(\psi^l, s, u, v) = u { D_t} - c_1 s + c_2 Y +
\left(v -c_2\right) q_o(v) ~~~~
\end{equation}
\begin{equation}
- ~ c(Y-q_o(v))^+, \label{eq:formRTCE} ~~~
\end{equation}\noindent which achieves the optimum at:
\begin{equation}
~~~~v^* = \left\{
\begin{array}{ll}
{{\gamma_o}(c_1 - c_p) \over {1+\gamma_o}} & {~~~~ \text{if}~~ Y \ge
q_o({{\gamma_o}(c_1 - c_p) \over {1+\gamma_o}})} \cr { {\gamma_o}
c_2 \over {1+\gamma_o}} & {~~~~ \text{if}~~ Y < q_o({{\gamma_o} c_2
\over {1+\gamma_o}})} \cr {q_o^{-1}(Y)} & {~~~~ \text{o.w.} ,}
\end{array}\right.
~~~~~~~~~~~~
\end{equation}\noindent which yields the real-time pricing policy $\mathbf{\tilde {\vartheta}}_{s,u}$
in the elastic case.
\end{proof}
\textbf{Remarks}: Note that the first case, i.e., $Y  \ge
q_o({{\gamma_o}(c_1 - c_p) \over {1+\gamma_o}})$, points to a case
with energy surplus, i.e., there is more energy supply than the
total demand from traditional and opportunistic users, whereas the
second case is tied to a case with energy deficit. Then, it is
natural that the real-time price in the first case depends on the
penalty $c_p$ in canceling a scheduled generation; and that the
real-time price in the second case is a function of the cost of fast
start-up generation. Note also that, in both cases, when the
opportunistic users become increasingly elastic, i.e.,
$\gamma_o\rightarrow-\infty$, the real-time price progressively
decreases to the minimum allowable prices, i.e., $c_1-c_p$ and
$c_2$, respectively. This monotonic behavior of the real-time price
with respect to increasing elasticity comes at no surprise, since,
as $\gamma_o\rightarrow-\infty$, the average opportunistic user
demand $q_o(v)\rightarrow 0 $, i.e., the opportunistic energy users
become more and more thrifty. Therefore, the scheduler must offer
power at increasingly cheaper prices, up to the lowest possible
price, to maximize profit.

Having established an approximate real-time scheduling policy for
both kind of opportunistic energy users, an approximate day-ahead
schedule can be obtained by solving the following single-stage
optimization:
\begin{eqnarray}
\mathbf{\mathcal{\tilde  P}^{DA}_{non-pst}}: \mathop {\max
}\limits_{S,u} {\mathbf{E}_{W}\mathbf{E}_{D_{t}}^{u} \left[
R^l(\psi^l, s, u, \mathbf{\tilde{\vartheta}}_{s,u}(\psi^l)\right]}.
\label{eq:formDA3}
\end{eqnarray}
where $R^l$ is given by (\ref{eq:9}) and the expectations depend on
the exact stochastic models assumed for wind generation and
traditional users' energy demand.

\begin{proposition} \label{lm:Uopt} The optimal decision of $\mathbf{\mathcal{\tilde  P}^{DA}_{non-pst}}$ is
given by
\begin{eqnarray}\nonumber u^* &=& \left\{
\begin{array}{ll}
{u_{cap} } & { \text{if}~~ -1 \le \gamma_{t} < 0}  \cr {{{\gamma_t}
\over {1+\gamma_t}} c_1} & { \text{if} ~~\gamma_t < -1,}
\end{array}\right. \\
\nonumber S^* &=& \mathop {\arg\max }\limits_{S} \left\{ (c_2-c_1)S
+ K \mathbf{E}_{W,D}^v \left[ (v-c_2) D_o  \right. \right. \\
\nonumber &&\left. \left. -~ c \mathbf{1}_{{B}}( S  - K(\alpha_t
{u^*}^{\gamma_t} + \varepsilon_t +D_o - W))\right] \right\}.
\end{eqnarray}
\end{proposition} \label{prop:3}
\begin{proof} We first show that $\mathbf{\tilde \vartheta}_{s,u}$
depend on $S$ and $u$ only through $s - \mathbf{E}_{D_t}^u[D_t]$.
For convenience, define:
\begin{equation}
\nonumber s' = s - \mathbf{E}_{D_t}^u[D_t].
\end{equation}
Since $Y = s'+ W - \varepsilon_t$, it is clear from Proposition
\ref{prop:1} and Proposition \ref{prop:2} that the real-time pricing
policy $ \mathbf{\tilde \vartheta}_{s,u}$ depends on the day-ahead
decision only through $s'$. We denote this policy as $
\mathbf{\tilde \vartheta}_{s'}:(W,\varepsilon_t)\rightarrow v$. With
this insight, by using the change of variable technique in
(\ref{eq:formRT0}), the objective function of
$\mathbf{\mathcal{\tilde P}^{DA}_{non-pst}}$ can be rewritten as
\begin{equation*}\label{eq:obj1}
\mathbf{E}_{W} \mathbf{E}^u_{D_t}\left[ R^l (\psi^l, s', u, \tilde
\vartheta_{s'}(W,\varepsilon_t))\right]~~~~~~~~~~~~~~~~~~~~~~
\end{equation*}
\begin{equation}
~~~~= f_1(u) + f_2(s') + c_2 \mathbf{E}_W[W], \label{eq:formDA4}
\end{equation}
\noindent where
\begin{equation}
f_1(u)  \buildrel \Delta \over = \alpha_t (u-c_1) u^{\gamma_t},
~~~~~~~~~~~~~~~~~~~~~~~~~~~~~~
\end{equation}
\begin{equation*}
f_2(s')  \buildrel \Delta \over = \mathbf{E}_{W,\varepsilon_t}
\mathbf{E}^{\tilde \vartheta_{s'}(W,\varepsilon_t)}_{D_o} \big[
(\tilde \vartheta_{s'}(W,\varepsilon_t)-c_2)D_o  ~~~~~~~~
\end{equation*}
\begin{equation}
~~~~~~~~~~~~~~~ - c \mathbf{1}_B\left(s' - \varepsilon_t - D_o + W
\right) \big] + (c_2 - c_1) s'.
\end{equation}
Let $\mathcal{F}$ denote the solution space for the objective
function of $\mathbf{\mathcal{\tilde P}^{DA}_{non-pst}}$ defined in
(\ref{eq:formDA4}). Thus
\begin{eqnarray}
\mathcal{F}&=&\{(u,s'); u\ge 0, s'\ge -\alpha_t u^{\gamma_t}\}.
\end{eqnarray}
%
%
It can be verified that $u^*$ defined in the proposition statement
maximizes $f_1(u)$. Define $s'_0 = - \alpha_t {u^*}^{\gamma_t}$, and
let $s'^*$ maximize $f_2(s')$. If we show that $(u^*,s'^*)$ belongs
to the solution space $\mathcal{F}$, then $(u^*,s'^*)$ optimizes the
day-ahead scheduling problem in (\ref{eq:formDA4}). Since $u^*\ge
0$, it is now sufficient to show that $s'^* \ge s'_0$. A sufficient
condition to establish this is given by
\begin{equation}
f_2(s') \le f_2(s'_0), ~\forall~ s' \le s'_0 \label{eq:Sufficond}.
\end{equation}
\noindent
Under condition $\mathbf{A}$, wind energy is not sufficient to meet
the total energy demand when day-ahead price is $u^*$, thus:
\begin{equation}
W < (-s'_0 + \varepsilon_t) + D_o.
\end{equation}
Therefore,
\begin{equation}
W + s' < \varepsilon_t + D_o, ~\forall~ s' \le s'_0.
\end{equation}
It follows that 1)  $W + s < D_t + D_o$, i.e., there is no scheduled
energy surplus, thus $\mathbf{1}_B = 0$ in (\ref{eq:formDA4}); 2)
Using the preceding statement, recalling the definition of $Y$, we
see that $Y<q_o(\frac{\gamma_o c_2}{1+\gamma_o})$. Thus, from
Proposition~\ref{prop:2}, the optimal real-time price
$\vartheta_{s'}(W,\varepsilon_t)$ turns out to be a constant
${{\gamma_o} c_2 \over {1+\gamma_o}}$, i.e., independent of the
system state and the day-ahead decisions, when the opportunistic
users are relatively elastic. Also, for the relatively inelastic
case, we know from Proposition~\ref{prop:1} that the optimal
real-time price is a constant $v_{cap}$. Letting $v_0(\gamma_o)$
denote this constant real-time price for both the elastic and
inelastic cases, respectively, we have
\begin{equation}
f_2(s') = (c_2 - c_1) s' + (v_0(\gamma_o)-c_2) q_o(v_0(\gamma_o)),
~\forall~ s' \le s'_0.
\end{equation}
Therefore, $f_2(s') \le f_2(s'_0), ~\forall~ s' \le s'_0$, and
$(u^*,s'^*)$ indeed lies in the feasible region $\mathcal{F}$ and
hence optimizes the day-ahead scheduling problem in
(\ref{eq:formDA4}). The optimal day-ahead decision, $S^*$, can now
be computed using $s'^*$ and $u^*$.
\end{proof}

\begin{corollary} When the non-persistent opportunistic energy users
are relatively inelastic, the complete two-timescale scheduling
decision is given by:
\begin{eqnarray*}
u^* = \left\{
\begin{array}{ll}
{u_{cap} } & { \text{if}~~ -1 \le \gamma_{t} < 0}  \cr {{{\gamma_t}
\over {1+\gamma_t}} c_1} & { \text{if} ~~\gamma_t < -1,}
\end{array}\right. ~~~~~~~
\end{eqnarray*}
\begin{equation*}
S^*= K \left( \kappa_1 E_o v_{cap}^{\gamma_o}  + \alpha_t
{u^*}^{\gamma_t} - \mathbf{F}^{-1}_Z\left(c_p / c \right)\right),~~~
\end{equation*}  \begin{equation*}
\tilde \vartheta_{s,u}(\psi^l) = v_{cap}, ~~~~~~~
~~~~~~~~~~~~~~~~~~~~~~~~~~~~~~
\end{equation*}\noindent where $\mathbf{F}_Z^{-1}(\cdot)$ denotes the inverse of the
CDF of $Z \buildrel \Delta \over = W-\varepsilon_t$.
\end{corollary}

\section{Dispatch and Scheduling with Persistent Opportunistic Energy Users}

We now study multi-timescale dispatch and scheduling when the
opportunistic users are {\em persistent}. Simply put, a persistent
opportunistic user waits
in the system for a new real-time price, in the next $T_2$-slot, if
the current real-time price is not acceptable. We assume that the
opportunistic users are persistent across \textit{both} $T_2$ and
$T_1$ slots, and that the opportunistic energy users that arrived in
the day leave the system at the end of the day.

%

Due to the persistent nature of the opportunistic users, scheduling
decisions in both $T_2$ and $T_1$-slots affect the system trajectory
and hence scheduling decisions in future time-slots, across both
timescales. Thus, the scheduling problem involves hierarchically
structured control \cite{Mahmoud77}, with the hierarchy defined
across timescales. With this insight, we treat the scheduling
problem as a {\em multi-timescale Markov decision process} (MMDP)
\cite{Chang03} where decisions made in the higher level affects both
the state transition dynamics and the decision process at the lower
level, while decisions at the lower level affect only the decisions
made at the upper level. The multi-timescale dispatch and scheduling
problem at hand is particularly unique in the following sense: the
two timescales \textit{do not overlap}, since the upper level
decisions (day-ahead) are made in \textit{non real-time}. Thus the
upper level scheduler does not have any direct observation of the
effect it has on the lower level system dynamics, until the horizon,
and make decisions solely based on stochastic understanding of the
behavior of the lower level process. These properties make the
two-timescale scheduling problem, with persistent users, uniquely
challenging. We now describe the problem in detail.

\begin{figure*}[hb]
\vspace*{2pt} \hrulefill
\begin{eqnarray}\label{eq:expectations}
E_{P^u_{m+1}}^{\{\psi^u_m, a_m^u\}}(.)&=&E_{P^l_{2,m}}^{P^l_{1,m}=P^u_m}E_{P^l_{3,m}}^{P^l_{2,m}}\ldots E_{P^l_{K,m}}^{P^l_{K-1,m}}E_{P^u_{m+1}}^{P^l_{K,m}}(.)\nonumber\\
E_{P^l_{k+1,m}}^{P^l_{k,m}}(.)&=&E_{W_{k,m}}E_{D_{t_{k,m}}}E_{N_{k,m}}E_{P^l_{k+1,m}}^{\{W_{k,m},D_{t_{k,m}},N_{k,m},P^l_{k,m}\}}(.)\nonumber\\
E_{P^l_{k+1,m}}^{\{W_{k,m},D_{t_{k,m}},N_{k,m},P^l_{k,m}\}}(.)&=&\sum_{{P^l_{k+1,m}}=0}^{N_{k,m}+P^l_{k,m}}
\left(
\begin{array}{cc}
P^l_{k,m}+N_{k,m}\\
P^l_{k+1,m}
\end{array}\right)
(1-\Pi_{v_{k,m}})^{P^l_{k+1,m}}(\Pi_{v_{k,m}})^{(N_{k,m}+P^l_{k,m}-P^l_{k+1,m})}(.)\nonumber\\
\Pi_{v_{k,m}}&=&P(V\le v_{k,m})\nonumber\\
v_{k,m}&=&\zeta(W_{k,m},D_{t_{k,m}},P^l_{k,m})\nonumber\\
E^{\{N_{k,m},P^l_{k,m},
v_{k,n}\}}_{N_{a_{k,n}}}(.)&=&\sum_{N_{a_{k,m}}=0}^{N_{k,m}+P^l_{k,m}}
\left(
\begin{array}{cc}
P^l_{k,m}+N_{k,m}\\
N_{a_{k,m}}
\end{array}\right)
(\Pi_{v_{k,m}})^{N_{a_{k,m}}}(1-\Pi_{v_{k,m}})^{(N_{k,m}+P^l_{k,m}-N_{a_{k,m}})}(.)
\end{eqnarray}
\end{figure*}

For day-ahead scheduling, the scheduler decides the energy dispatch
$S_m$ and the retail price $u_m$ for the $m$th $T_1$-slot in the
next day. Recall that the day-ahead scheduler has an accurate
forecast of the expected amount of wind generation $\theta_m$. We
define the system state, $\psi^u_m$, corresponding to the $m$th
$T_1$-slot in the next day as $\psi^u_m=\{\theta_m,P^u_m\}$, where
$P^u_m$ denotes the number of persistent opportunistic users carried
over from the $(m-1)$th $T_1$-slot. During real-time scheduling, the
scheduler, in each $T_2$-slot, has knowledge of the wind generation
and the demand from traditional energy users, along with the number
of persistent opportunistic users carried over from previous
$T_2$-slot. Based on this information, it must decide a real-time
price $v_{k,m}$ for the $k$th $T_2$-slot in the $m$th $T_1$-slot,
i.e., the $(k,m)$th $T_2$-slot). We define the observable
(observable in real time) state of the system in $(k,m)$th
$T_2$-slot as $\psi^l_{k,m}=\{W_{k,m},{D_t}_{k,m},P^l_{k,m}\}$,
where $W_{k,m}$ denotes the wind generation in the $k$th $T_2$-slot
in the $m$th $T_1$-slot , $D_{t_{k,m}}$ is the energy demand from
traditional energy users in the $(k,m)$th $T_2$-slot, $P^l_{k,m}$
denotes the number of persistent opportunistic energy users carried
over from the previous $T_2$-slot to the $(k,m)$th slot. Having
explicitly defined the states of the system for the day-ahead
scheduling and real-time scheduling, we introduce the optimality
equations next. With $\vec{X}_m=[X_{1,m}, \ldots, X_{K,m}]$, we have
\begin{eqnarray}
\lefteqn{V_m^u({\psi^u}_m)=}\nonumber\\
&&\max_{s_m,u_m}\Big\{E_{\vec{W}_m,\vec{D}_{t_m}}\Big[\max_{\vec{v}_m}\big\{E_{\vec{P}_m}\nonumber\\
&&\big[R_{k,m}^{l}(\{{\psi^u}_m,\psi^l_{k,m}\},s_m,u_m,v_{k,m})\big]\nonumber\\
&&+E_{P^u_{m+1}=P^l_{K,m}}V_{m+1}^u(\psi_{m+1}^u)\big\}\Big]\Big\}
\end{eqnarray}
As noted earlier, this is a MMDP over a finite horizon, with
$s_m,u_m$ being the upper level (slower timescale) decisions and
$v_{k,m}$ being the lower level decisions.

To mitigate the complexity of the MMDP problem, next we exploit the
structural properties of the multi-scale dispatch and scheduling
problem and recast it as a classic Markov decision process (MDP)
(e.g., \cite{puterman94}).

\begin{proposition}
With appropriately defined immediate reward $R_m^u$ and action space
$a^u_m$, the two-level scheduling problem can be written as a
classic MDP at the slower time-scale, as below:
\begin{eqnarray}
V_m^u(\psi^u_m)&=&\max_{a^u_m=\{s_m,u_m,\zeta_m\}}\Big\{R_m^u(\psi_m^u,a^u_m)\nonumber\\
&&+E_{P^u_{m+1}}^{\{\psi^u_m,a^u_m\}}V_{m+1}^u(\psi^u_{m+1})\Big\}.\nonumber
\end{eqnarray}
\end{proposition}

We now proceed to discuss the transformation of the two-level
scheduling problem from a MMDP to a classic MDP. In the two-level
scheduling problem, recall that the lower level decisions are
essentially the mapping from the realizations of wind energy,
traditional users' energy demand and persistent opportunistic users
to the real-time price, i.e., $ \zeta_{k,m}:
\{W_{k,m},{D_t}_{k,m},P^l_{k,m}\}\rightarrow v_{k,m}$. Consider a
stationary real-time pricing rule within each $T_1$ slot, i.e.
$\zeta_{1,m}=\zeta_{2,m}\ldots \zeta_{K,m}$, and denote this
stationary mapping by $\zeta_m$. A key step is to view $\zeta_m$ as
an action at the day-ahead scheduling level, in addition to actions
$s_m,u_m$. With this insight, we can simplify the MMDP into a
classic MDP. We elaborate further on this below.

The expected net reward from slot $m\in\{1,\ldots,D\}$ until slot
$M$ in the day-ahead market is thus given by the Bellman equation:
\begin{eqnarray}
V_m^u(\psi^u_m)&=&\max_{a^u_m=\{s_m,u_m,\zeta_m\}}\Big\{R_m^u(\psi_m^u,a^u_m)\nonumber\\
&&+E_{P^u_{m+1}}^{\{\psi^u_m,a^u_m\}}V_{m+1}^u(\psi^u_{m+1})\Big\}
\end{eqnarray}
where, recall, $E_x^y$ refers to the expectation over $x$
conditioned on $y$. The expectations used in the MDP formulation are
explicitly provided in (\ref{eq:expectations}).
The terminal reward, $V_M^u$, is given by
\begin{eqnarray}
V_{M}^u(\psi^u_M)&=&\max_{a_M^u}R_M^u(\psi^u_M,a_M^u)
\end{eqnarray}
Note that the immediate reward corresponding to $m$th $T_1$-slot is
a function of the realized values of wind generation mean (that is
accurately forcast 24 hours ahead) and the number of persistent
opportunistic users carried over from previous $T_1$-slot.
We now proceed to explicitly characterize the immediate reward
$R^u_m$, for $m\in\{1,\ldots M\}$:
\begin{eqnarray}
\lefteqn{R_m^u(\psi^u_m,a^u_m)}\nonumber\\
&=&V_{k,m}^l(\{\theta_m,P^l_{k,m}\},a^u_m)|_{k=1}
\end{eqnarray}
where, $P^l_{1,m}=P^u_m$ by definition, and for
$k\in\{1,\ldots,K\}$,
\begin{eqnarray}
\lefteqn{V_{k,m}^l(\{\theta_m,P^l_{k,m}\},a^u_m=\{s_m,u_m,\zeta_m\})}\nonumber\\
&=&E_{W_{k,m}}^{\theta_m}
E_{D_{t_{k,m}}}^{u_m}R_{k,m}^l(\psi^l_{k,m}=\{W_{k,m},D_{t_{k,m}},P^l_{k,m}\},a_m^u)\nonumber\\
&&+E_{P^{l}_{k+1,m}}^{P^l_{k,m}}V_{k+1,m}^l(\{\theta_m,P^l_{k+1,m}\},a^u_m),
\end{eqnarray}
where $V_{K,m}$ is given by
\begin{eqnarray}
\lefteqn{V_{K,m}^l(\{\theta_m,P^l_{K,m}\},a^u_m)}\nonumber\\
&=&E_{W_{K,m}}^{\theta_m}
E_{D_{t_{K,m}}}^{u_m}R_{K,m}^l(\psi^l_{K,m},a_m^u).
\end{eqnarray}
Note that the quantities $R^l_{k,m}$ and $V^l_{k,m}$ can be regarded
as the immediate reward and the net reward at the lower level MDP,
if the problem is viewed as an MMDP. The quantity $R^l_{k,m}$ is a
function of the realizations of the wind generation, the demand from
traditional energy users and the number of persistent opportunistic
energy users carried over from previous $T_2$-slot. More
specifically, we have that
\begin{eqnarray}
\lefteqn{R_{k,m}^l(\psi^l_{k,m},a_m^u)}\nonumber\\
&=&u_m D_{t_{k,n}}+E_{N_{k,m}}E^{\{N_{k,m},P^l_{k,m},v_{k,m}\}}_{N_{a_{k,m}}}\Big[v_{k,m}D_{o_{k,m}}\nonumber\\
&&+\textbf{1}_A(-c_p s_m)+\textbf{1}_B(-\epsilon_{k,m}c_p-(s_m-\epsilon_{k,m})c_1)\nonumber\\
&&+\textbf{1}_C(-c_1 s_m+c_2\epsilon_{k,m})\Big]
\end{eqnarray}
where $N_{k,m}$ denotes the number of opportunistic users arriving
at the $(k,m)$th slot, and $N_{a_{k,m}}$ denotes the number of
opportuinstic users that become active in slot $(k,m)$.
$D_{o_{k,m}}$ denotes the energy consumption by the opportunistic
users in slot $(k,m)$. Recall that $E_o$ denotes the energy demand
per \textit{active} opportunistic energy user. Thus, we have
\begin{eqnarray}
D_{o_{k,m}}&=&N_{a_{k,m}}E_o
\end{eqnarray}
The quantity $\epsilon_{k,m}$ denotes the surplus energy, given by
\begin{eqnarray}
\epsilon_{k,m}&=&W_{k,m}+s_m-(D_{t_{k,m}}+ D_{o_{k,m}}).
\end{eqnarray}
Recall, the real-time price $v_{k,m}$ is given by the mapping
$\zeta_m$ as
\begin{eqnarray}
v_{k,m}&=&\zeta_m(\{W_{k,m},D_{t_{k,m}},P^l_{k,m}\}).
\end{eqnarray}
The indicator functions $\mathbf{1}_A$, $\mathbf{1}_B$ and
$\mathbf{1}_C$ correspond to the various deficit/surplus energy
states, and were explicitly defined in Section~III.

Summarizing, we have shown that the two-timescale dispatch and
scheduling problem with persistent users can be recast as an MDP
with continuous state and action spaces. Using appropriate
discretization techniques, we can reformulate it as a classic
discrete state and action space MDP, which can be solved optimally
or near-optimally using various solution techniques available in the
literature \cite{shwartz05}.


\section{Numerical Results}
For concreteness, we now study, via numerical experiments, the
performance of the proposed multi-timescale dispatch and scheduling
policy. First, we define profit margin as the ratio of the average
net profit to the sales in a $T_2$-slot. We assume the opportunistic
users are non-persistent.
\begin{figure}[htb]
\begin{center}
\includegraphics
[width=0.44\textwidth]{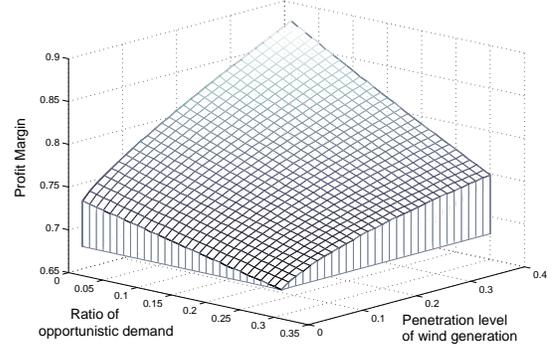}
\end{center}
\caption{Profit margin at various penetration levels of wind
generation and opportunistic demand ($\gamma_o = -2$)}
\label{fig:Profit1}
\end{figure}

In Fig.~\ref{fig:Profit1}, the profit margin is plotted against
various values of the penetration of opportunistic users and
penetration of wind energy in the system with two-timescale
scheduling. We use price elasticity of $-0.5$ for traditional users
and $-2$ for opportunistic users. As expected, the penetration of
wind energy increases the profit margin since wind energy is
harvested `cost-free'. In contrast, as the demand of penetration of
opportunistic energy users increases, the  profit margin
\textit{decreases}. To explain this, we first identify the two
principal events that lead to losses in the system operator's
profit: (a) The penalty ($c_p$) when a fraction of the dispatched
conventional generation is reverted; (b) the loss of revenue when
the dispatched energy $S$ is insufficient to meet the total energy
needs and the operator purchases energy from fast-start up
generators at higher cost $c_2$. The uncertainty in opportunistic
energy users' demand is one of the factors that could lead to either
of these events. An over-estimation of this demand in the day-ahead
schedule leads to event (a), while an underestimation of this demand
leads to event (b). In our experiments, since the demand from
traditional users is deterministic, the higher the demand from
opportunistic users, the higher the uncertainties on the demand side
and hence the higher the losses at the system operator. This insight
further underscores the need for efficient pricing mechanisms that
intelligently tackles the ever-increasing uncertainties in the power
system.

\begin{figure}[htb]
\begin{center}
\includegraphics
[width=0.4\textwidth]{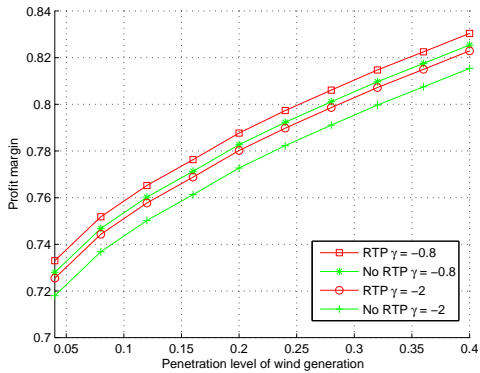}
\end{center}
\caption{Profit margin with various elasticity ($r_u$ =
20\%)}\label{fig:curve2}
\end{figure}

Fig.~\ref{fig:curve2} compares the profit margins with and without
multi-timescale pricing for various values of wind energy
penetration levels, i.e., the ratio of the wind energy generation to
the total energy. In the benchmark system, without multi-timescale
price, \textit{all} the users exhibit traditional response to prices
with the same price elasticity $\gamma$. The prices are optimized by
taking into consideration the statistics of the traditional users'
demand and wind generation. Considering the multi-timescale
scheduling system, for fair comparison, we assume that the
opportunistic users have the same elasticity as the traditional
users in the system. We compare the two systems for different values
of $\gamma$. As expected, for a fixed wind penetration level, for
the same value of $\gamma$, the profit margin is higher when
real-time pricing is employed to mitigate the uncertainties in
demand and supply. Also, note that, as the price elasticity
increases, the opportunistic users exhibit increasingly thrifty
behavior essentially reducing the profit margin for the operator. In
addition, as expected, the profit margin increases with wind
penetration with or without real-time pricing since wind energy is
assumed to be cost-free.

%

To summarize, numerical results suggest that two-time scale
scheduling effectively addresses the volatility of energy generation
and the uncertainties in the demand from opportunistic users.
Additional insights include the following: the profit margin of
system operators increases with the penetration level of wind and
decreases with demand from opportunistic users and that the
elasticity of the opportunistic users plays a major role in the
power system design.

\section{Conclusion}

Wind generation is among the renewable resources that has most
variability and uncertainty, and exhibits multi-level dynamics
across time.  Aiming to tackle the challenge of integrating volatile
wind generation into the bulk power grid, we study multiple
timescale dispatch and scheduling, for a smart grid model, via
day-ahead scheduling and real-time scheduling. In day-ahead
scheduling, with the statistical information on wind generation and
energy demands, we characterize the optimal procurement of the
energy supply and the day-ahead retail price for the traditional
energy users; in real-time scheduling, with the realization of wind
generation and the load of traditional energy users, we optimize
real-time prices to manage the opportunistic energy users so as to
achieve system-wide reliability. More specifically, when the
opportunistic users are non-persistent, we obtain closed-form
solutions to the multi-scale scheduling problem. For the persistent
case, we treat the scheduling problem as a multi-timescale Markov
decision process, and then we show that it can be recast,
explicitly, as a classic Markov decision process.

We believe that the studies we initiated here on multi-timescale
dispatch and scheduling for integrating volatile renewal energy into
smart grids, scratch only the tip of the iceberg. There are still
many questions remaining open to improve the penetration of
renewable energy into power grids, and we are currently
investigating these issues along this avenue.

\bibliographystyle{IEEETr}
\bibliography{Reference}

\end{document}